\documentclass[letterpaper, 10pt, conference]{ieeeconf}
\IEEEoverridecommandlockouts       
\usepackage{amssymb,amsmath,amsfonts,amsthm}
\usepackage{algorithm,algorithmic}
\usepackage{cite}
\usepackage{float}
\usepackage{epsfig}
\usepackage{graphicx}
\usepackage{graphics}
\usepackage{subfigure}
\usepackage{url}
\usepackage{xspace}
\usepackage[usenames,dvipsnames]{color}
\usepackage{soul}
\usepackage{mathtools}

\floatstyle{ruled}
\newfloat{model}{H}{mod}
\floatname{model}{\footnotesize Model}
\newfloat{notatio}{H}{not}
\floatname{notatio}{\footnotesize Notation}

\newenvironment{varalgorithm}[1]
  {\algorithm\renewcommand{\thealgorithm}{#1}}
  {\endalgorithm}

\newenvironment{list4}{
	\begin{list}{$\bullet$}{%
			\setlength{\itemsep}{0.05cm}
			\setlength{\labelsep}{0.2cm}
			\setlength{\labelwidth}{0.3cm}
			\setlength{\parsep}{0in} 
			\setlength{\parskip}{0in}
			\setlength{\topsep}{0in} 
			\setlength{\partopsep}{0in}
			\setlength{\leftmargin}{0.16in}}}
	{\end{list}}

\newenvironment{list4a}{
	\begin{list}{$\bullet$}{%
			\setlength{\itemsep}{0.05cm}
			\setlength{\labelsep}{0.2cm}
			\setlength{\labelwidth}{0.3cm}
			\setlength{\parsep}{0in} 
			\setlength{\parskip}{0in}
			\setlength{\topsep}{0in} 
			\setlength{\partopsep}{0in}
			\setlength{\leftmargin}{0.16in}}}
	{\end{list}}

\usepackage{url,changebar,bm,xspace,dsfont}
\let\mathbb=\mathds 
\def\day{\mathrm{day}} 

\newtheorem{theorem}{Theorem}

\newtheorem{assum}{Assumption}

\newtheorem{remark}{Remark}

\usepackage{amsmath,amsfonts,amssymb,amscd}
\usepackage{capt-of}
\usepackage{color}
\usepackage{cite}
\usepackage{verbatim}
\usepackage{hyperref}

\hypersetup{
    colorlinks = true,
    linkcolor = [rgb]{0,0,1},
    anchorcolor = [rgb]{0,0,1},
    citecolor = [rgb]{0.9,0.5,0},
    filecolor = [rgb]{0,0,1},
    urlcolor = [rgb]{0.9,0.5,0},
    bookmarksopen = true,
    bookmarksnumbered = true,
    breaklinks = true,
    linktocpage,
    colorlinks = true,
    linkcolor = [rgb]{0.2,0.6,0.2},
    urlcolor  = [rgb]{0,0,1},
    citecolor = [rgb]{0.9,0.5,0},
    anchorcolor = [rgb]{0.2,0.6,0.2},
}

\begin{document}

\title{\LARGE \bf Distributed Optimization via Gradient Descent with \\ Event-Triggered Zooming over Quantized Communication}

\author{Apostolos~I.~Rikos, Wei~Jiang, Themistoklis~Charalambous, and Karl~H.~Johansson
\thanks{Apostolos~I.~Rikos is with the Department of Electrical and Computer Engineering, Division of Systems Engineering, Boston University, Boston, MA 02215, US. E-mail: {\tt arikos@bu.edu}.}
\thanks{Wei Jiang resides in Hong Kong, China. Email: {\tt wjiang.lab@gmail.com}.}
\thanks{T.~Charalambous is with the Department of Electrical and Computer Engineering, School of Engineering, University of Cyprus, 1678 Nicosia, Cyprus.  
He is also with the Department of Electrical Engineering and Automation, School of Electrical Engineering, Aalto University, Espoo, Finland. 
E-mail:{\tt charalambous.themistoklis@ucy.ac.cy}.
}
\thanks{K.~H.~Johansson is with the Division of Decision and Control Systems, KTH Royal Institute of Technology, SE-100 44 Stockholm, Sweden. He is also affiliated with Digital Futures. E-mail: {\tt kallej@kth.se}.}
\thanks{Part of this work was supported by the Knut and Alice Wallenberg Foundation, the Swedish Research Council, and the Swedish Foundation for Strategic Research. 
The work of T. Charalambous was partly supported by the European Research Council (ERC) Consolidator Grant MINERVA (Grant agreement No. 101044629).
}
}

\maketitle
\pagestyle{empty}

%
%
%
%
\begin{abstract}
In this paper, we study unconstrained distributed optimization strongly convex problems, in which the exchange of information in the network is captured by a directed graph topology over digital channels that have limited capacity (and hence information should be quantized). 
Distributed methods in which nodes use quantized communication yield a solution at the proximity of the optimal solution, hence reaching an error floor that depends on the quantization level used; the finer the quantization the lower the error floor. 
However, it is not possible to determine in advance the optimal quantization level that ensures specific performance guarantees (such as achieving an error floor below a predefined threshold).
Choosing a very small quantization level that would guarantee the desired performance, requires {information} packets of very large size, which is not desirable (could increase the probability of packet losses, increase delays, etc) and often not feasible due to the limited capacity of the channels available.
In order to obtain a communication-efficient distributed solution and a sufficiently close proximity to the optimal solution, we propose a quantized distributed optimization algorithm that converges in a finite number of steps and is able to adjust the quantization level accordingly. 
The proposed solution uses a finite-time distributed optimization protocol to find a solution to the problem for a given quantization level in a finite number of steps and keeps refining the quantization level until the difference in the solution between two successive solutions with different quantization levels is below a certain pre-specified threshold. 
Therefore, the proposed algorithm progressively refines the quantization level, thus eventually achieving low error floor with a reduced communication burden.
The performance gains of the proposed algorithm are demonstrated via illustrative examples.
\end{abstract}


%
%
%
%
\section{Introduction}\label{sec:intro}

The problem of distributed optimization has become increasingly important in recent years due to the rise of large-scale machine learning \cite{2020:Nedich}, control \cite{SEYBOTH:2013}, and other data-driven applications \cite{2018:Stich_Jaggi} that involve massive amounts of data. 

Most distributed optimization algorithms in current literature assume that nodes exchange real valued messages of infinite precision \cite{2009:Nedic_Optim, 2018:Khan_AB, 2020:Vivek_Salapaka, 2021:Tiancheng_Uribe, 2022:Wei_Themis_CDC}. 
In distributed computing settings, nodes typically communicate with each other over a network that has limited communication bandwidth and latency. 
This means that exchanging messages with infinite precision can be impractical or even impossible. 
More specifically, the assumption of infinite-capacity communication channels is unrealistic because it requires the ability to transmit an infinite number of bits per second. 
Additionally, most distributed algorithms assume the transmission of rational numbers, which however, is only possible over infinite-capacity communication channels. 

In order to alleviate the aforementioned limiting assumption, researchers have focused on the scenario where nodes are exchanging quantized\footnote{Quantization is the process of mapping input values from a large set (often a continuous set) to output values in a (countable) smaller set. In quantization, nodes compress  (i.e., quantize) their value (of their state or any other stored information), so that they can represent it with a few bits and then transmit it through the channel.} messages \cite{2014:Peng_Yiguan, 2016:Huang_Xiao, 2017:Ye_Zeilinger_Jones, 2018:Huaqing_Xie, 2021:Doan_Romberg, 2020:Magnusson_NaLi, 2021:Jhunjhunwala_Eldar, 2021:Kajiyama_Takai, 2022:Liu_Daniel,  2019:Koloskova_Jaggi, 2019:Basu_Diggavi, 2020:Jadbabaie_Federated, 2020:Li_Chi}. 
This may lead to a solution to the proximity of the optimal solution that depends on the utilized quantization level. 
However, most of the proposed {works} are mainly quantizing values of an asymptotic coordination algorithm. 
As a consequence, they are only able to exhibit asymptotic convergence to a solution in the proximity of the optimal solution. 


A recent work \cite{2023:Rikos_Johan_IFAC} proposed a finite-time communication-efficient algorithm for distributed optimization. However, it is not obvious how coarse/fine the quantization should be. If it is too coarse, the solution to the optimization may lead to an error floor that is considerably large and hence, unacceptable (for the considered application). If it is too fine, then larger packets are needed for communication (which means that the overall system may experience delays, more packet losses, etc). Since the exact solution is not known \emph{a priori} though, it is not possible to know whether the quantization level chosen is sufficient. 


\textbf{Main Contributions.} 
In this paper, we present a novel distributed optimization algorithm aimed at addressing the challenge of quantization level tuning. 
Our proposed algorithm extends the quantized distributed optimization method in~\cite{2023:Rikos_Johan_IFAC} (which converges to an approximate solution within a finite number of iterations).  
Our key contribution is a strategy that dynamically adjusts the quantization level based on the comparison of error floors resulting from different quantization levels. 
The proposed strategy allows us to assess the satisfaction of the obtained solution, even in the absence of knowledge about the optimal solution. 
Our key contributions are the following. 
\\ \noindent 
\textbf{A.} We present a distributed optimization algorithm that leverages on gradient descent and fosters efficient communication among nodes through the use of quantized messages; see Algorithm~\ref{alg1}. 
Our algorithm operates by comparing solutions obtained with different quantization levels. 
If these solutions exceed a predefined threshold, we continue to refine the quantization level, otherwise, we terminate its operation; see for example Fig.~\ref{plot_optimality}. 
While we cannot directly enforce the exact desired accuracy, our algorithm can attain a desired level of accuracy through the selection of an appropriate threshold. 
For example, by setting the threshold in the order of $10^{-7}$, we can guarantee an error floor as low as $10^{-6}$.
Remarkably, with each iteration of the optimization process, the quantization granularity becomes finer, and the initial conditions approach the vicinity of the optimal state.  
This behavior resembles a distributed zooming process over the optimization region.
\\ \noindent
\textbf{B.} We validate the performance of our proposed algorithm through illustrative examples, demonstrating its effectiveness in terms of communication efficiency and the computation of optimal solutions; see Section~\ref{sec:results}. 
The achieved improvement in communication efficiency is substantial and holds practical significance; see Remark~\ref{commun_contr_remark}.

\section{NOTATION AND PRELIMINARIES}\label{sec:preliminaries}

\textbf{Notions.} 
The sets of real, rational, integer and natural numbers are denoted by $ \mathbb{R}, \mathbb{Q}, \mathbb{Z}$ and $\mathbb{N}$, respectively. 
The symbol $\mathbb{Z}_{\geq 0}$ denotes the set of nonnegative integer numbers. 
The symbol $\mathbb{R}_{\geq 0}$ denotes the set of nonnegative real numbers. 
The symbol $\mathbb{R}^n_{\geq 0}$ denotes the nonnegative orthant of the $n$-dimensional real space $\mathbb{R}^n$. 
Matrices are denoted with capital letters (e.g., $A$), and vectors with small letters (e.g., $x$). 
The transpose of matrix $A$ and vector $x$ are denoted as $A^\top$, $x^\top$, respectively. 
For any real number $a \in \mathbb{R}$, the floor $\lfloor a \rfloor$ denotes the greatest integer less than or equal to $a$ while the ceiling $\lceil a \rceil$ denotes the least integer greater than or equal to $a$. 
For any matrix $A \in \mathbb{R}^{n \times n}$, the $a_{ij}$ denotes the entry in row $i$ and column $j$. 
By $\mathbb{1}$, we denote the all-ones vector and by $\mathbb{I}$ the identity matrix of appropriate dimensions. 
By $\| \cdot \|$, we denote the Euclidean norm of a vector. 


\textbf{Graph Theory.} 
The communication network is captured by a directed graph (digraph) defined as $\mathcal{G} = (\mathcal{V}, \mathcal{E})$. 
This digraph consists of $n$ ($n \geq 2$) nodes communicating only with their immediate neighbors, and is static (i.e., it does not change over time). 
In $\mathcal{G}$, the set of nodes is denoted as $\mathcal{V} =  \{ v_1, v_2, ..., v_n \}$, and the set of edges as {$\mathcal{E} \subseteq \mathcal{V} \times \mathcal{V} \setminus \{ (v_i, v_i) \ | \ v_i \in \mathcal{V} \}$ (note that self-edges are excluded)}. 
The cardinality of the sets of nodes, edges are denoted as $| \mathcal{V} |  = n$, $| \mathcal{E} | = m$, respectively. 
A directed edge from node $v_i$ to node $v_l$ is denoted by $(v_l, v_i) \in \mathcal{E}$, and captures the fact that node $v_l$ can receive information from node $v_i$ (but not the other way around). 
The subset of nodes that can directly transmit information to node $v_i$ is called the set of in-neighbors of $v_i$ and is represented by $\mathcal{N}_i^- = \{ v_j \in \mathcal{V} \; | \; (v_i, v_j)\in \mathcal{E}\}$. 
The subset of nodes that can directly receive information from nodes $v_i$ is called the set of out-neighbors of $v_i$ and is represented by $\mathcal{N}_i^+ = \{ v_l \in \mathcal{V} \; | \; (v_l, v_i)\in \mathcal{E}\}$. 
The \textit{in-degree}, and \textit{out-degree} of $v_j$ and is denoted by $\mathcal{D}_i^- = | \mathcal{N}_i^- |$, $\mathcal{D}_i^+ = | \mathcal{N}_i^+ |$, respectively. 
The diameter $D$ of a digraph is the longest shortest path between any two nodes $v_l, v_i \in \mathcal{V}$. 
A directed \textit{path} from $v_i$ to $v_l$ of length $t$ exists if we can find a sequence of nodes $i \equiv l_0,l_1, \dots, l_t \equiv l$ such that $(l_{\tau+1},l_{\tau}) \in \mathcal{E}$ for $ \tau = 0, 1, \dots , t-1$. 
A digraph is \textit{strongly connected} if there exists a directed path from every node $v_i$ to every node $v_l$, for every $v_i, v_l \in \mathcal{V}$. 

\vspace{.2cm}

\textbf{Node Operation.} 
{Each node $v_i \in \mathcal{V}$ executes a distributed optimization algorithm and a distributed coordination algorithm. 
For the optimization algorithm (see Algorithm~\ref{alg1} (GraDeZoQuC) below) at each time step $k$, each node $v_i$ maintains 
\begin{itemize}
    \item its local estimate variable $x_i^{[k]} \in \mathbb{Q}$ (used to calculate the optimal solution),\
    \item $\gamma_\beta$ which is the time step during which nodes have converged to a neighborhood of the optimal solution, 
    \item the set $S_i$ which is used to store the $\gamma_\beta$, 
    \item the variable $\text{ind}_i$ (used as an indicator of the length of the set $S_i$), 
    \item the variable $\text{flag}_i$ (used to decide whether to terminate the optimization algorithm operation). 
\end{itemize}
For the coordination algorithm (Algorithm~\ref{alg2} (FiTQuAC) below) at each time step $k$, each node $v_i$ maintains 
\begin{itemize}
    \item the stopping variables $M_i$, $m_i \in \mathbb{N}$ (used to determine whether convergence has been achieved), and
    \item the variables $y_i \in \mathbb{Q}$, $c^y_i, c_i^z \in \mathbb{Z}$, and $z_i \in \mathbb{Q}$, (used to communicate with other nodes by either transmitting or receiving messages). 
\end{itemize}}


\textbf{Asymmetric Quantizers.} 
Quantization is a strategy that lessens the number of bits needed to represent information. 
It is used to compress data before transmission, thus reducing the amount of bandwidth required to transmit messages, and increasing power and computation efficiency. 
Quantization is mainly used to describe communication constraints and imperfect information exchanges between nodes such as in wireless communication systems, distributed control systems, and sensor networks. 
The three main types of quantizers are (i) asymmetric, (ii) uniform, and (iii) logarithmic \cite{2019:Wei_Johansson}. 
In this paper we rely on asymmetric quantizers in order to reduce the required communication bandwidth (but our results can also be extended to logarithmic and uniform quantizers). 
Asymmetric quantizers are defined as 
\begin{equation}\label{asy_quant_defn}
    q_{\Delta}^a(\xi) = \Bigl \lfloor \frac{\xi}{\Delta} \Bigr \rfloor , 
\end{equation}
where $\Delta \in \mathbb{Q}$ is the quantization level, $\xi \in \mathbb{R}$ is the value to be quantized, and $q_{\Delta}^a(\xi) \in \mathbb{Q}$ is the quantized version of $\xi$ with quantization level $\Delta$ (note that the superscript ``$a$'' indicates that the quantizer is asymmetric.). 


The $\max$-consensus algorithm converges to the maximum value among all nodes in a finite number of steps $s_m \leq D$, where $D$ is the network diameter (see, \cite[Theorem 5.4]{2013:Giannini}). 
Similar results hold for the $\min$-consensus algorithm.

%
%
%
%
\section{Problem Formulation}\label{sec:probForm}

Let us consider a distributed network modeled as a digraph $\mathcal{G} = (\mathcal{V}, \mathcal{E})$ with $n  = | \mathcal{V} |$ nodes. 
We assume that 
each node $v_i$ is endowed with a local cost function $f_i(x): \mathbb{R}^p \mapsto \mathbb{R}$ only known to itself, and
communication channels among nodes have limited capacity and as a result the exact states cannot be communicated if they are irrational. In other words, only quantized values can be transmitted/communicated and thus $x$ can take values that can be expressed as rational numbers.
 
In this paper we aim to develop a distributed algorithm which allows nodes, despite the communication limitations, to cooperatively solve approximately the following optimization problem, herein called \textbf{P1}: 
\begin{subequations}
\begin{align}
\min_{x\in \mathcal{X}}~ & F(x_1, x_2, ..., x_n) \equiv \sum_{i=1}^n f_i(x_i), \label{Global_cost_function}  \\
\text{s.t.}~ & x_i = x_j, \forall v_i, v_j \in \mathcal{V}, \label{constr_same_x}  \\
       & x_i^{[0]} \in \mathcal{X} \subset \mathbb{Q}_{\geq 0}, \forall v_i \in \mathcal{V}, \label{constr_x_in_X} \\
       & \text{nodes communicate with quantized values, } \label{constr_quant}  \\
       & \text{if} \ \| f_i(x_i^{[\gamma_{\beta-1}]}) - f_i(x_i^{[\gamma_\beta]}) \| \leq \varepsilon_s, \ \forall v_i \in \mathcal{V}, \nonumber \\
       & \text{for any} \ \varepsilon_s > 0, \text{then terminate operation, } \label{constr_stop} 
\end{align} 
\end{subequations}

where $\beta \in \mathbb{N}$, $\gamma_\beta$ is the optimization convergence point for which we have $f_i(x_i^{[1 + \gamma_{\beta}]}) = f_i(x_i^{[\gamma_\beta]})$ $\forall v_i \in \mathcal{V}$, $\mathcal{X}$ is the set of feasible values of parameter $x$, and $x^*$ is the optimal solution of the optimization problem.  
Eq.~\eqref{Global_cost_function} means that we aim to minimize the global cost function which is defined as the sum of the local cost functions in the network. 
Eq.~\eqref{constr_same_x} means that nodes need to calculate equal optimal solutions. 
Eq.~\eqref{constr_x_in_X} means that the initial estimations of nodes belong in a common set. 
Note that it is not necessary for the initial values of nodes to be rational numbers, i.e., $x_i^{[0]} \in \mathcal{X} \subset \mathbb{Q}_{\geq 0}$. 
However, nodes can generate a quantized version of their initial states by utilizing the Asymmetric Quantizer presented in Section~\ref{sec:preliminaries}.
Eq.~\eqref{constr_quant} means that nodes are transmitting and receiving quantized values with their neighbors since communication channels among nodes have limited bandwidth. 
Eq.~\eqref{constr_stop} means that nodes are tracking the improvement of their local cost function between two consecutive convergence points $\gamma_{\beta + 1}$ and $\gamma_{\beta}$. 
If the improvement of the local cost function of every node is less than a predefined threshold $\varepsilon_s$, then they decide to stop their operation in a distributed way. 



\begin{remark}
It will be shown later that our algorithm converges to a neighborhood of the optimal solution due to the quantized communication between nodes (see \eqref{constr_quant}). 
Therefore, with $\gamma_\beta$ we denote the time step for which all nodes have converged to this neighborhood (i.e., it is the optimization convergence point), and for this reason $f_i(x_i^{[1 + \gamma_{\beta}]}) = f_i(x_i^{[\gamma_\beta]}), \forall v_i \in \mathcal{V}$. 
\end{remark}

%
%
%
%
\section{Distributed Optimization with Zooming over Quantized Communication}\label{sec:distr_grad_zoom_quant}

In this section, we present a distributed algorithm which solves problem \textbf{P1} described in Section~\ref{sec:probForm}. 
Before presenting the operation of our proposed algorithm, we make the following assumptions which are necessary for the development of our results. 

\begin{assum}\label{str_conn}
The communication network (described as a digraph) $\mathcal{G}$ is \textit{strongly connected}. 
\end{assum}

\begin{assum}\label{assup_convex}
    For every node $v_i$, the local cost function $f_i(x)$ is smooth and strongly convex. 
    This means that for every node $v_i$, for every $x_1, x_2 \in \mathcal{X}$, 
    \begin{list4}
        \item there exists positive constant $L_i$ such that  \begin{equation}\label{lipschitz_defn}
            \| \nabla f_i(x_1) - \nabla f_i(x_2) \|_2 \leq L_i \| x_1 - x_2 \|_2, 
        \end{equation}
        \item there exists positive constant $\mu_i$ such that 
        \begin{equation}\label{str_conv_defn}
             f_i(x_2) \geq f_i(x_1) + \nabla f_i(x_1)^\top (x_2 - x_1) + \frac{\mu_i}{2} \| x_2 - x_1 \|_2^2. 
        \end{equation}
    \end{list4}
        
\noindent This means that the Lipschitz-continuity and strong-convexity constants of the global cost function $F$ (see \eqref{Global_cost_function}) are $L$ $\mu$, {defined as $L = \max \{ L_i\}$, and $\mu = \min \{ \mu_i\}$.} 
\end{assum}

\begin{assum}\label{digr_diam}
The diameter $D$ (or an upper bound) is known to every node $v_i$ in the network. 
\end{assum}



Assumption~\ref{str_conn} is a necessary condition so that information from each node can reach every other node in the network, thus all nodes to be able to calculate the optimal solution $x^*$ of $P1$. 
Assumption~\ref{assup_convex} is the Lipschitz-continuity condition in \eqref{lipschitz_defn}, and strong-convexity condition in \eqref{str_conv_defn}. 
Lipschitz-continuity is a standard assumption in distributed first-order optimization problems (see \cite{2018:Xie, 2018:Li_Quannan}) and guarantees (i) the existence of the solution $x^*$, and (ii) that nodes are able to calculate the global optimal minimizer $x^*$ for \eqref{Global_cost_function}. 
Strong-convexity is useful for guaranteeing (i) linear convergence rate, and (ii) that the global function $F$ has no more than one global minimum. 
Assumption~\ref{digr_diam} allows each node $v_i \in \mathcal{V}$ to determine whether calculation of a solution $x_i$ that fulfills \eqref{constr_same_x} has been achieved in a distributed manner. 


{The intuition of Algorithm~\ref{alg1} (GraDeZoQuC) is the following. 
\\ 
\indent
\textit{Initialization.} Each node $v_i$ maintains an estimate of the optimal solution $x_i^{[0]}$, the desired quantization level $\Delta$, and the refinement constant $c_r$ which is used to refine the quantization level. 
Quantization level (i) is the same for every node, (ii) allows quantized communication between nodes, and (iii) determines the desired level of precision of the solution. 
Additionally, each node initializes a set $S_i$. 
This set serves as a repository for storing the time steps during which nodes have collectively calculated the neighborhood of the optimal solution according to the utilized quantization level $\Delta$.
More specifically, Algorithm~\ref{alg1} converges to a neighborhood of the optimal solution due to the quantized communication between nodes. 
Each node $v_i$ stores in $S_i$ the optimization time step during which this neighborhood has been reached. 
\\ \indent
\textit{Iteration.} At each time step $k$, each node $v_i$: 
\begin{itemize}
    \item Updates the estimate of the optimal solution $x_i^{[k+\frac{1}{2}]}$ by performing a gradient descent step towards the negative direction the node's gradient; see Iteration step~$1$.
    \item Utilizes Algorithm~\ref{alg2} (FiTQuAC); see Iteration step~$2$.
    Algorithm~\ref{alg2} (details of its operation are presented below) allows each node to fulfill \eqref{constr_quant}, and to calculate in finite time an estimate of the optimal solution $x_i^{[k+1]}$ that fulfills \eqref{constr_same_x}. 
    \item Checks if the calculated estimate of the optimal solution $x_i^{[k+1]}$ is the same as the previous optimization step $x_i^{[k]}$; see Iteration step~$3$.
    \item If the above condition holds, then nodes have reached a neighborhood of the optimal solution which depends on the utilized quantization level (i.e., they reached the optimization convergence point for the current quantization level).
    In this case, node $v_i$ stores the corresponding time step $\gamma_{\beta} = k$ at the set $S_i$; see Iteration steps~$3a$, $3b$.
    \item Checks if the difference between the value of its local function at the current optimization convergence point $f_i(x_i^{[\gamma_{\beta}]})$ and the value of its local function at the previous optimization convergence point $f_i(x_i^{[\gamma_{\beta-1}]})$ is less than a given threshold $\varepsilon_s$; see Iteration step~$3c$.
    \item If the above condition holds, it sets its voting variable equal to $0$ (otherwise it sets it to $1$). 
    Then nodes are performing a max-Consensus protocol to decide whether they will continue the operation of Algorithm~\ref{alg1}; see Iteration step~$3d$.
    The main idea for executing max-Consensus is that if every node finds that the difference between $f_i(x_i^{[\gamma_{\beta}]})$ and $f_i(x_i^{[\gamma_{\beta-1}]})$ is less than $\varepsilon_s$ (signaling convergence) then nodes opt to halt their operation. 
    \item After executing max-Consensus, if at least one node detects that the difference exceeds $\varepsilon_s$ (indicating a lack of convergence) then nodes utilize the refinement constant $c_r$ to adjust the quantization level and repeat the algorithm's operation accordingly, otherwise the operation is terminated; see Iteration step~$3e$.
\end{itemize}


Algorithm~\ref{alg2} (FiTQuAC) allows each node to be able to calculate the quantized average of each node’s estimate in finite time by processing and transmitting quantized messages, with precision determined by the quantization level. 
FAQuA algorithm utilizes (i) asymmetric quantization, (ii) quantized averaging, and (iii) a stopping strategy. 
{The intuition of Algorithm~\ref{alg2} (FiTQuAC) is the following. 
Initially, each node $v_i$ uses an asymmetric quantizer to quantize its state; see Initialization-step~$2$. 
Then, at each time step $\eta$ each node $v_i$: 
\begin{itemize}
    \item Splits the $y_i$ into $z_i$ equal pieces (the value of some pieces might be greater than others by one); see Iteration-steps~$4.1$, $4.2$. 
    \item Transmits each piece to a randomly selected out-neighbor or to itself; see Iteration-step~$4.3$.  
    \item Receives the pieces transmitted from its in-neighbors, sums them with $y_i$ and $z_i$, and repeats the operation; see Iteration-step~$4.4$. 
\end{itemize}
Finally, every $D$ time steps, each node $v_i$ performs in parallel a max-consensus and a min-consensus operation; see Iteration-steps~$1$, $2$, $5$. 
If the results of the max-consensus and min-consensus have a difference less or equal to one, each node $v_i$ (i) scales the solution according to the quantization level, (ii) stops the operation of Algorithm~\ref{alg2}, (iii) uses the value $x_i^{[k+1]}$ to continue the operation of Algorithm~\ref{alg1}.} 
Algorithm~\ref{alg2} converges in finite time according to \cite[Theorem~$1$]{2021:Rikos_Hadj_Splitting_Autom}.
It is important to note here that Algorithm~\ref{alg2} (FiTQuAC) runs between every two consecutive optimization steps $k$ and $k + 1$ of Algorithm~\ref{alg1} (GraDeZoQuC) (for this reason it uses a different time index $\lambda$ and not $k$ as Algorithm~\ref{alg1}). 

Our proposed algorithm is detailed below as Algorithm~\ref{alg1}.  

\begin{varalgorithm}{1}
\caption{Gradient Descent with Zoomed Quantized Communication (GraDeZoQuC)}
\textbf{Input:} A strongly connected directed graph $\mathcal{G}$ with $n = |\mathcal{V}|$ nodes and $m = |\mathcal{E}|$ edges. 
Static step-size $\alpha \in \mathbb{R}$, digraph diameter $D$, initial value $x_i^{[0]}$, local cost function $f_i$, error bound $\varepsilon_s$, quantization level $\Delta \in \mathbb{Q}$, refinement constant $c_r \in \mathbb{N}$, for every node $v_j \in \mathcal{V}$. 
Assumptions~\ref{str_conn}, \ref{assup_convex}, \ref{digr_diam} hold. 
\\
\textbf{Initialization:} Each node $v_i \in \mathcal{V}$ sets $\text{ind}_i = 0$, $\beta = \text{ind}_i$, $S_i = \{ 0 \}$. \\ 
\textbf{Iteration:} For $k = 0,1,2,\dots$, each node $v_i \in \mathcal{V}$ does the following: 
\begin{list4}
\item[1)] $x_i^{[k+\frac{1}{2}]} =  x_i^{[k]} - \alpha \nabla f_i(x_i^{[k]})$; 
\item[2)] $x_i^{[k+1]} = $ Algorithm~\ref{alg2}($x_i^{[k+\frac{1}{2}]}, D, \Delta $); 
\item[3)] \textbf{if} $x_i^{[k+1]} = x_i^{[k]}$, \textbf{then}
\begin{list4a}
\item[$3a)$] set $\text{ind}_i = \text{ind}_i + 1$, $\beta = \text{ind}_i$, $\gamma_{\beta} = k$; 
\item[$3b)$] set $S_i = S_i \cup \{ \gamma_{\beta} \}$; 
\item[$3c)$] \textbf{if} $\| f_i(x_i^{[\gamma_{\beta-1}]}) - f_i(x_i^{[\gamma_\beta]}) \| \leq \varepsilon_s$, \textbf{then} set $\text{vot}_i = 0$;\\
 \textbf{else} set $\text{vot}_i = 1$; 
\item[$3d)$] $\text{flag}_i$ = max - Consensus ($\text{vot}_i$); 
\item[$3e)$] \textbf{if} $\text{flag}_i = 0$ \textbf{then} terminate operation; \\
 \textbf{else} set $\Delta = \Delta / c_r$ and go to Step~$1$;
\end{list4a} 
\end{list4} 
\textbf{Output:} Each node $v_i \in \mathcal{V}$ calculates $x_i^*$ which solves problem \textbf{P1} in Section~\ref{sec:probForm}. 
\label{alg1} 
\end{varalgorithm}

\noindent
\vspace{-0.3cm}    
\begin{varalgorithm}{2}
\caption{Finite-Time Quantized Average Consensus (FiTQuAC)}
\textbf{Input:} $x_i^{[k+\frac{1}{2}]}, D, \Delta$. 
\\
\textbf{Initialization:} Each node $v_i \in \mathcal{V}$ does the following: 
\begin{list4}
\item[$1)$] Assigns probability $b_{li}$ to each out-neigbor $v_l \in \mathcal{N}^+_i \cup \{v_i\}$, as follows
\begin{align*}
b_{li} = \left\{ \begin{array}{ll}
         \frac{1}{1 + \mathcal{D}_i^+}, & \mbox{if $l = i$ or $v_{l} \in \mathcal{N}_i^+$,} \\
         0, & \mbox{if $l \neq i$ and $v_{l} \notin \mathcal{N}_i^+$;}\end{array} \right. 
\end{align*} 
\item[$2)$] sets $z_i = 2$, $y_i = 2 \  q_{\Delta}^a(x_i^{[k+\frac{1}{2}]})$ (see \eqref{asy_quant_defn}); 
\end{list4} 
\textbf{Iteration:} For $\lambda = 1,2,\dots$, each node $v_i \in \mathcal{V}$, does: 
\begin{list4a}
\item[$1)$] \textbf{if} $\lambda \mod (D) = 1$ \textbf{then} $M_i = \lceil y_i  / z_i \rceil$, $m_i = \lfloor y_i / z_i \rfloor$; 
\item[$2)$] broadcasts $M_i$, $m_i$ to every $v_{l} \in \mathcal{N}_i^+$; receives $M_j$, $m_j$ from every $v_{j} \in \mathcal{N}_i^-$; sets $M_i = \max_{v_{j} \in \mathcal{N}_i^-\cup \{ v_i \}} M_j$, \\ $m_i = \min_{v_{j} \in \mathcal{N}_i^-\cup \{ v_i \}} m_j$; 
\item[$3)$] sets $c_i^z = z_i$; 
\item[$4)$] \textbf{while} $c_i^z > 1$ \textbf{do} 
\begin{list4a}
\item[$4.1)$] $c^y_i = \lfloor y_{i} \  / \  z_{i} \rfloor$; 
\item[$4.2)$] sets $y_{i} = y_{i} - c^y_i$, $z_{i} = z_{i} - 1$, and $c_i^z = c_i^z - 1$; 
\item[$4.3)$] transmits $c^y_i$ to randomly chosen out-neighbor $v_l \in \mathcal{N}^+_i \cup \{v_i\}$ according to $b_{li}$; 
\item[$4.4)$] receives $c^y_j$ from $v_j \in \mathcal{N}_i^-$ and sets 
\begin{align}
y_i & = y_i + \sum_{j=1}^{n} w^{[r]}_{\lambda,ij} \ c^y_{j} \ , \\
z_i & = z_i + \sum_{j=1}^{n} w^{[r]}_{\lambda,ij} \ ,
\end{align}
where $w^{[r]}_{\lambda,ij} = 1$ when node $v_i$ receives $c^y_{i}$, $1$ from $v_j$ at time step $\lambda$ (otherwise $w^{[r]}_{\lambda,ij} = 0$ and $v_i$ receives no message at time step $\lambda$ from $v_j$);
\end{list4a} 
\item[$5)$] \textbf{if} $\lambda \mod D = 0$ \textbf{and} $M_i - m_i \leq 1$ \textbf{then} sets $x_i^{[k+1]} = m_i \Delta$ and stops operation. 
\end{list4a}
\textbf{Output:} $x_i^{[k+1]}$. 
\label{alg2} 
\end{varalgorithm}

\subsection{Convergence of Algorithm~\ref{alg1}}\label{ConvADMMAlg}
We now analyze the convergence time of Algorithm~\ref{alg1} via the following theorem. 

\begin{theorem}\label{converge_Alg1}
Under Assumptions~\ref{str_conn}--\ref{digr_diam}, when the step-size $ \alpha $ satisfies $ \alpha \in (\frac{n(\mu + L)}{4\mu L}, \frac{2n}{\mu + L}) $ and $ \delta \in (0,\frac{n[4 \alpha \mu L -n(\mu + L) ]}{2 \alpha [n(\mu + L)-2  \alpha \mu L]} ) $ where  $ L =  \max \{ L_i\} ,  \mu = \min \{ \mu_i\}   $,
	Algorithm~\ref{alg1} generates a sequence of points $ \{x^{[k]}\} $ (i.e., the variable $x_i^{[k]}$ of each node $v_i \in \mathcal{V}$) which satisfies
	\begin{align}\label{linear_convergence}
	\| \hat{x}^{[k+1]} - x^{*}\|^2 
	<
	\vartheta\| \hat{x}^{[k]} - x^{*}  \|^2   + \mathcal{O}(\Delta^2) ,
	\end{align}
	where $ \Delta $ is the quantizer and
	\begin{subequations}\label{throrem1}
		\begin{align}
		\vartheta := & 2(1+\frac{\alpha\delta}{n} )(1- \frac{2\alpha\mu L}{n(\mu + L)} ) \in (0,1), \label{throrem1_a}\\
		 \mathcal{O}(\Delta^2) 
		 =&(8
		 +32 n^2\hat{\alpha}^2 L^2 + \frac{32 n^2\hat{\alpha} L^2}{\delta}) \Delta^2. \label{throrem1_b} 
		\end{align}
	\end{subequations}
\end{theorem}

\begin{proof}
The proof follows directly from the proof of~ \cite[Theorem~1]{2023:Rikos_Johan_IFAC}, with the difference that the process is restarted under some condition (eq.~\eqref{constr_stop}). The details are omitted due to space limitations. 
\end{proof}

\begin{remark}[Convergence Precision]
The focus of our convergence analysis in Theorem~\ref{converge_Alg1} is on the optimization steps performed during the operation of Algorithm~\ref{alg1}. 
As stated, an additional term $\mathcal{O}(\Delta^2)$ appears in \eqref{linear_convergence}. 
This term affects the precision of the calculated optimal solution. 
While some distributed quantized algorithms in the literature exhibit exact convergence to the optimal solution (e.g., see \cite{2014:Peng_Yiguan, 2021:Kajiyama_Takai}), our Algorithm~\ref{alg1} adopts an adaptive quantization level to balance communication efficiency and convergence precision. 
However, by setting $\varepsilon_s = 0$ during Initialization, Algorithm~\ref{alg1} can be adjusted to converge to the \textit{exact} optimal solution $x^*$ (by refining the quantization level infinitely often). 
This characteristic is highly important in scenarios where higher precision is crucial. 
Specifically, Algorithm~\ref{alg1} is able to adjust to specific application requirements by performing a trade-off between communication efficiency and convergence precision. 
Furthermore, it is worth noting that Algorithm~\ref{alg1} offers distinct advantages particularly in scenarios where communication efficiency is a priority while maintaining satisfactory convergence precision in various applications. 
As will be shown in Section~\ref{sec:results}, the operational advantages of Algorithm~\ref{alg1} are evident, making it a valuable tool in distributed optimization tasks. 
\end{remark}

\section{Simulation Results}\label{sec:results}




In this section, we present simulation results in order to demonstrate the operation of Algorithm~\ref{alg1} and its potential advantages. 
More specifically: 
\\ \noindent
\textbf{A.} We focus on a random digraph of $20$ nodes and show how the nodes’ states converge to the optimal solution (see Fig.~\ref{plot_optimality}). 
Furthermore, we analyze how the event-triggered zooming (i) leads to a more precise calculation of the optimal solution, and (ii) allows nodes to terminate their operation.  
\\ \noindent
\textbf{B.} We compare the operation of Algorithm~\ref{alg1} against existing algorithms in the literature, and we emphasize on the introduced improvements (see Fig.~\ref{plot_comparisons}). 

For both cases \textbf{A.} and \textbf{B.} each node $v_i$ is endowed with a local cost function $f_i(x) = \frac{1}{2} \beta_i (x - x_0)^2$. 
This cost function is smooth and strongly convex. 
Furthermore, for $f_i(x)$ we have that (i) $\beta_i$ is initialized as a random integer between $1$ and $5$ for each node in the network (and characterizes the cost sensitivity of node $v_i$), and (ii) $x_0$ is initialized as a random integer between $1$ and $5$ (and represents the demand of node $v_i$). 


\noindent
\textbf{A. Operation over a random digraph of $20$ nodes.} 
In Fig.~\ref{plot_optimality}, we demonstrate our algorithm over a randomly generated digraph consisted of $20$ nodes. 
For each node $v_i$ we have $\alpha = 0.12$, $x_i^{[0]} \in [1, 5]$, $\varepsilon_s = 0.003$, $\Delta = 0.001$, $c_r=10$. 
In Fig.~\ref{plot_optimality}, we plot the error $e^{[k]}$ in a logarithmic scale against the number of iterations. 
The error $e^{[k]}$ is defined as 
\begin{equation}\label{eq:distance_optimal}
    e^{[k]} = \sqrt{ \sum_{j=1}^n \frac{(x_j^{[k]} - x^*)^2}{(x_j^{[0]} - x^*)^2} } , 
\end{equation}
where $x^*$ is the optimal solution of the problem \textbf{P1}. 

\begin{figure}[t]
    \centering
    \includegraphics[width=.88\linewidth]{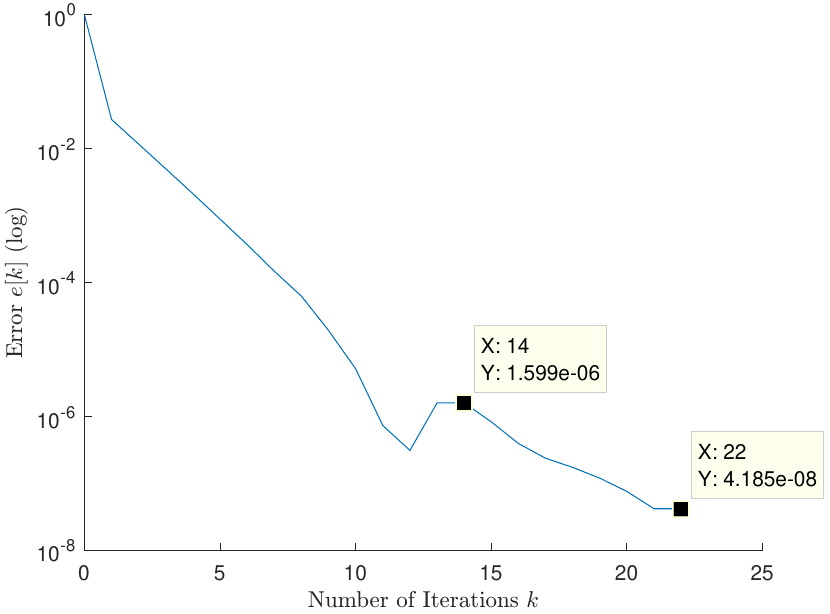}
    \caption{Execution of Algorithm~\ref{alg1} over a random digraph of $20$ nodes.}
    \label{plot_optimality}
\end{figure}

In Fig.~\ref{plot_optimality} we can see that our algorithm is able to converge to the optimal solution. 
Furthermore, let us focus at time steps $k= 13, 14$, and $k= 21, 22$. 
At time steps $k= 13, 14$ we have that the condition in Iteration Step~$3$ holds (i.e., $x_i^{[13]} = x_i^{[14]}$ for every $v_i \in \mathcal{V}$), and $e^{[13]} = e^{[14]}$. 
Therefore, during time step $14$, nodes are checking the overall improvement of their local cost functions (i.e., Iteration Step~$3c$). 
Since this condition does not hold for at least one node, they decide to refine the quantization level (i.e., set $\Delta = \Delta / 10 = 0.0001$), and continue executing Algorithm~\ref{alg1}. 
At time steps, $k = 14, ..., 21$, nodes are able to approximate the optimal solution with more precision than before since the precision depends on the quantization level (as we showed in Theorem~\ref{converge_Alg1}).  
At time steps $k= 21, 22$ we have that the condition in Iteration Step~$3$ holds again. 
However, during time step $22$ the overall improvement of every nodes' local cost function is less than the given threshold $\epsilon_s$, i.e., $\| f_i(x_i^{[14]}) - f_i(x_i^{[22]}) \| \leq \varepsilon_s$, for every $v_i \in \mathcal{V}$ (see Iteration Step~$3c$). 
As a result, nodes decide to terminate the operation at time step $k=22$ (see Iteration Step~$3e$). 
Note here that a choice of a smaller $\epsilon_s$ may lead nodes to refine again the quantization level. 
This refinement (i.e, $\Delta \leq 0.00001$) will allow them to approximate the optimal solution with even higher precision. 


\noindent
\textbf{B. Comparison with current literature.} 
In Fig.~\ref{plot_comparisons}, we compare the operation of Algorithm~\ref{alg1} against \cite{2022:Wei_Themis_CDC, 2023:Rikos_Johan_IFAC}. 
We plot the error $e^{[k]}$ defined in \eqref{eq:distance_optimal}. 
For the operation of the three algorithms, for each node $v_i$ we have $\alpha = 0.12$, $x_i^{[0]} \in [1, 5]$, $\varepsilon_s = 27 \cdot 10^{-7}$, $\Delta = 0.001$, $c_r= 10 $ (note that \cite{2022:Wei_Themis_CDC} is not utilizing $\varepsilon_s$, $\Delta$, $c_r$, and \cite{2023:Rikos_Johan_IFAC} is not utilizing $\varepsilon_s$, $c_r$).
Our comparisons focus on:
\\ \noindent
\textbf{B-A.} The convergence of Algorithm~\ref{alg1} compared to \cite{2022:Wei_Themis_CDC, 2023:Rikos_Johan_IFAC}. 
\\ \noindent
\textbf{B-B.} The required communication for convergence (in terms of bits per optimization step) of Algorithm~\ref{alg1} compared to \cite{2022:Wei_Themis_CDC, 2023:Rikos_Johan_IFAC}. 

\begin{figure}[t]
    \centering
    \includegraphics[width=.9\linewidth]{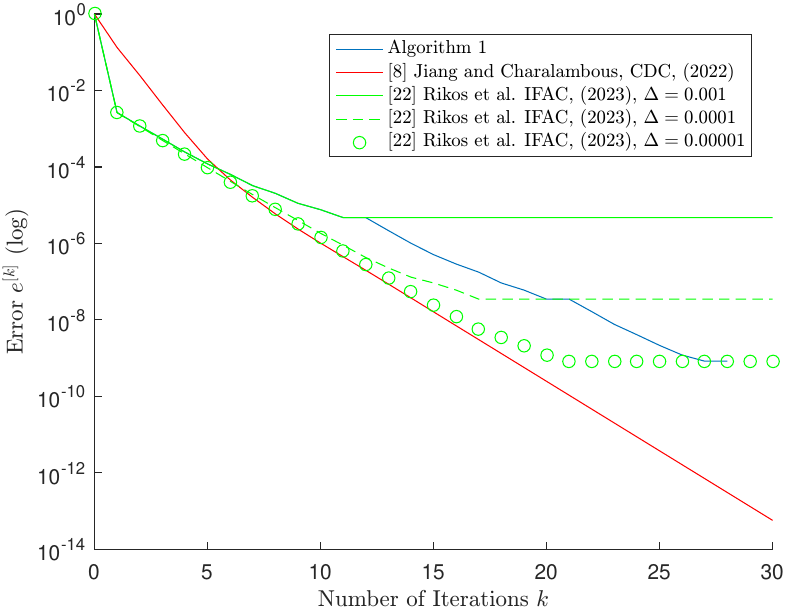}
    \caption{Comparison of Algorithm~\ref{alg1} against \cite{2022:Wei_Themis_CDC, 2023:Rikos_Johan_IFAC} over a random digraph of $20$ nodes.}
    \label{plot_comparisons} 
\end{figure}

\textbf{B-A} (Convergence). In Fig.~\ref{plot_comparisons} we can see that Algorithm~\ref{alg1} converges identically to \cite{2023:Rikos_Johan_IFAC} for optimization steps $k=0, ..., 12$. 
However, at time step $12$, each node refines the quantization level (because the condition at Iteration Step~$3c$ of Algorithm~\ref{alg1} does not hold for at least one node). 
In this case, for time steps $k > 12$ we can see that Algorithm~\ref{alg1} approximates the optimal solution with higher precision than \cite{2023:Rikos_Johan_IFAC}. 
This is mainly because \cite{2023:Rikos_Johan_IFAC} utilizes a static quantization level, which is not refined during the operation of the algorithm. 
Then, at time step $k = 21$, Algorithm~\ref{alg1} refines again the quantization level, obtaining an even more precise estimation of the optimal solution. 
However, at time step $k = 27$, we have that the condition at Iteration Step~$3c$ holds for every node and Algorithm~\ref{alg1} terminates its operation. 
Finally, in Fig.~\ref{plot_comparisons} we can see that \cite{2022:Wei_Themis_CDC} exhibits linear convergence rate and is the fastest among the three algorithms. 
However, during its operation, each node needs to form the Hankel matrix and perform additional computations when the matrix loses rank. 
This requires the exact values from each node. 
It means that nodes need to exchange messages of infinite capacity which is practically infeasible and imposes excessive communication requirements over the network. 
Therefore the main advantage of Algorithm~\ref{alg1} compared to \cite{2022:Wei_Themis_CDC}, is that nodes exchange quantized values guaranteeing efficient communication. 

\textbf{B-B} (Communication). In Fig.~\ref{plot_comparisons}, let us focus on comparing Algorithm~\ref{alg1} with \cite{2023:Rikos_Johan_IFAC} for $\Delta = 0.00001$ (see green circles line in Fig.~\ref{plot_comparisons}). 
Specifically, we will focus on the communication requirements (in terms of total number of bits and bits per optimization time step) for achieving the error $e^{[27]}$ for Algorithm~\ref{alg1} (which is the same as the error $e^{[21]}$ for the algorithm in \cite{2023:Rikos_Johan_IFAC}). 
The communication bits are calculated as the ceiling of the base-$2$ logarithm of the transmitted values. 
For example if node $v_i$ transmits the quantized value $\alpha$, then the number of bits it transmits is equal to $\lceil \log_2( a ) \rceil$. 
Note that comparing Algorithm~\ref{alg1} with \cite{2023:Rikos_Johan_IFAC} for $\Delta = 0.001$, and $\Delta = 0.0001$ can be shown identically. 
In Fig.~\ref{plot_comparisons}, we have that during the operation of \cite{2023:Rikos_Johan_IFAC} for $\Delta = 0.00001$, nodes are utilizing \textit{in total} $800754$ bits for communicating with their neighbors. 
This means that the average communication requirement for each node is $\frac{800754}{(20)(21)} = 1906.55$ bits per optimization time step (since the network consists of $20$ nodes which need $21$ iterations to converge). 
During the operation of Algorithm~\ref{alg1}, nodes are utilizing $\Delta = 0.001$ for steps $k = 0, ..., 12$, $\Delta = 0.0001$ for steps $k = 13, ..., 21$, and $\Delta = 0.00001$ for steps $k = 22, ..., 27$. 
For steps $k = 0, ..., 12$, nodes are utilizing \textit{in total} $195607$ bits for communicating with their neighbors. 
For steps $k = 12, ..., 21$, nodes are utilizing \textit{in total} $215635$ bits for communicating with their neighbors. 
For steps $k = 21, ..., 27$, nodes are utilizing \textit{in total} $201044$ bits for communicating with their neighbors. 
The total requirement of bits is $612286$ for $k=0, ..., 27$. 
This means that the average communication requirement for each node is $\frac{612286}{(20)(27)} = 1133.86$ bits per optimization time step. 
As a result, Algorithm~\ref{alg1}, is able to approximate the optimal solution with precision similar to \cite{2023:Rikos_Johan_IFAC} (for $\Delta = 0.00001$), but its communication requirements are significantly lower in terms of total number of bits and bits per optimization time step. 





\begin{remark}\label{commun_contr_remark}
    During the analysis in \textbf{B-B}, we have that Algorithm~\ref{alg1} requires less bits for communication compared to \cite{2023:Rikos_Johan_IFAC} (for $\Delta = 0.00001$) because nodes are utilizing a higher quantization level than \cite{2023:Rikos_Johan_IFAC} for optimization steps $k = 1, ..., 21$. 
    This means that nodes are utilizing less bits to quantize and transmit their states towards their neighboring nodes. 
    However, note here that during the operation of Algorithm~\ref{alg1} we can further improve communication efficiency by shifting the quantization basis after we refine the quantization step. 
    Shifting the quantization basis means changing the location of the quantization levels relative to the states of the nodes. 
    This can be done by adding/subtracting a constant value to the states of the nodes before quantization. 
    This constant value that we can subtract is equal to the optimal solution to which the states of the nodes have converged before refining the quantization level. 
    For example, in Fig.~\ref{plot_comparisons}, during optimization step $k=15$, node $v_i$ will quantize the state $x_i^{[15]} - x_i^{[\gamma_1]}$ (where $x_i^{[\gamma_1]}$ is equal to $x_i^{[12]}$). 
    This strategy increases even further communication efficiency since the states of the nodes can be represented using fewer bits without sacrificing the accuracy of the calculated optimal solutions during the optimization operation. 
    It will be further analyzed at an extended version of our paper. 
\end{remark}

%
%
%
%

\section{Conclusions}\label{sec:conclusions}


In this paper, we considered an unconstrained distributed strongly convex optimization problem, in which the exchange of information is done over digital channels that have limited capacity (and hence information should be quantized). 
%
We proposed a distributed algorithm that solves the problem with a solution at a close proximity to the optimal, by progressively refining the quantization level of a node, thus guaranteeing a certain error floor and more efficient communication (smaller packets/reduced number of bits).
A simple numerical example shows the performance of our proposed algorithm and highlights the benefits in terms of communication efficiency. More specifically, in the specific example it was shown that the number of bits needed is $\sim 25\%$ less when the quantization level is refined.
\bibliographystyle{IEEEtran}
\bibliography{bibliografia_consensus}

\end{document}